%% file: root.tex
\documentclass[conference]{IEEEtran}
\IEEEoverridecommandlockouts
\def\BibTeX{{\rm B\kern-.05em{\sc i\kern-.025em b}\kern-.08em
    T\kern-.1667em\lower.7ex\hbox{E}\kern-.125emX}}

\input{macros}

\begin{document}
\title{A Convex Optimization Framework for Computing Robustness Margins of Kalman Filters}

\author{
\begin{tabular}{cc}
Himanshu Prabhat & Raktim Bhattacharya\\
\small \texttt{himanshu.pr007@tamu.edu} & \small\texttt{raktim@tamu.edu}
\end{tabular}\\[1mm]
\small Aerospace Engineering, Texas A\&M University,\\
College Station, TX, 77843-3141.
}

\maketitle

\begin{abstract}
This paper proposes a novel convex optimization framework for designing robust Kalman filters that guarantee a user-specified steady-state error while maximizing process and sensor noise. The proposed framework simultaneously determines the Kalman gain and the robustness margin in terms of the process and sensor noise. This is the first paper to present such a joint formulation for Kalman filtering. The proposed methodology is validated through two distinct examples: the Clohessy-Wiltshire-Hill equations for a chaser spacecraft in an elliptical orbit and the longitudinal motion model of an F-16 aircraft. 
\end{abstract}

\section{Introduction}
The robustness margin in Kalman filtering is essential in quantifying the filter's performance under various uncertainties. It measures how well the filter copes with model uncertainties, deviations in noise characteristics, and environmental changes. A higher robustness margin reflects a Kalman filter's increased reliability and adaptability in dealing with system variations and uncertainties. Specifically, the upper bounds of process and sensor noise covariances are vital in assessing a Kalman filter's robustness. 

Consider a discrete-time dynamical system
\begin{subequations}
\begin{align}
\text{State Dynamics:} && \x_{k+1} &= \A\x_k + \B\w_k,\\
\text{Measurement Model:}&&\y_k &= \C\x_k + \n_k,
\end{align}
\elab{model}
\end{subequations}
where $\x_k\in\real^n$ is the state vector, $\w_k\in\real^m$ is the process noise, $\y_k\in\real^p$ is the measurement data, and $\n_k$ is the sensor noise.  We assume the process and the sensor noise to be zero mean Gaussian, i.e.,  $\Exp{\w_k} = \vo{0}$, $\Exp{\n_k} = \vo{0}$, $\Exp{\w_k\vo{w}^T_k} = \Q$, and $\Exp{\n_k\vo{n}^T_k} = \R$.

In conventional Kalman filtering, the system's model is used to predict the future uncertainty in the state, resulting in the \textit{prior} state uncertainty. When new measurements are available, they are combined with the prior to obtain a \textit{posterior} state uncertainty with minimum error variance, where the error is defined as the difference between the true state and the predicted state.

The minimum steady-state estimation error's covariance $\Sigmab_\infty$ for a given $\R$ and $\Q$ is given by the solution of the following algebraic Riccati equation 
\begin{multline*}
\A\Sigss\A^T + \B\Q\B^T - \Sigss  \\ -\A\Sigss\C^T(\C\Sigss\C^T + \R)^{-1}\C\Sigss\A^T = \vo{0},
\end{multline*}
assuming $(\A,\B\Q\B^T)$ is controllable and $(\C,\A)$ is observable. The steady-state estimation error worsens with increased process and sensor noise.

The sensor noise is primarily dictated by its internal characteristics, impacting its performance and cost. High-precision sensors with low noise (covariance) are typically more expensive, leading to a tradeoff driven by design economics.

The process noise represents the uncertainty in the model used to predict the system's state. This uncertainty could be due to factors like unmodeled dynamics, external disturbances, or simplifications in the model. A larger $\Q$ indicates higher uncertainty in the model-based predictions.

Significant research efforts have been dedicated to the realm of low-precision sparse sensing, with a specific focus on sensor scheduling and sensor selection \cite{Boyd2009,ZHANG2017202,CUMBO2021107830,Das2021OptimalSP,VD2021,VD2022,VD2024Hinf}. Earlier studies addressing optimal sensing architecture design assumed fixed sensor precision \cite{Boyd2009,ZHANG2017202,CUMBO2021107830}. The challenge of sparse sensing, coupled with sensor precision minimization, has been tackled within the Kalman filtering approach in \cite{Das2021OptimalSP} and the $\mathcal{H}_2/\mathcal{H}_{\infty}$ framework in \cite{VD2021,VD2022,VD2024Hinf}. Notably, these existing algorithms lack considerations for uncertainty in process noise statistics.

To address this gap, uncertainties in both process and sensor noise within the Kalman filtering framework have been approached using robust methods such as risk-sensitive filtering \cite{Speyer1974}, M-estimation \cite{Durovic1999}, and min-max formulation \cite{Hasen2014,Han2023DRK}. However, these approaches do not provide a mechanism for error budgeting, hindering the ability to assign user-specified tradeoffs between process and sensor noise.

Overall performance typically mandates specific error budgets for each subsystem when designing complex systems. In the context of Kalman filters, the main challenge lies in achieving these target error budgets by selecting sensors with the appropriate level of precision and identifying the maximum process noise that the filter can effectively tolerate. In cases where the system has a particularly stringent error budget, even the optimal error covariance, derived from a combination of assumed process and sensor noise characteristics, might not suffice. In general, increasing sensing precision increases tolerance to process noise. However, this tradeoff is difficult to determine for a given filter performance. This challenge becomes even more pronounced in large-scale distributed systems, where determining the most effective sensing architecture — sensor precision and placement — and ensuring maximum robustness against process noise is complex and non-trivial. In these scenarios, designers might inadvertently incorporate sensors with excessively high precision or position them in locations that do not significantly enhance the system's robustness to process noise. Such decisions lead to an increased system cost without yielding any tangible improvements in robustness. This paper addresses this issue by proposing a solution as a convex optimization formulation, offering an efficient approach to sensor selection and placement for increased filter robustness.

The main contribution of this paper is a novel convex optimization framework to design Kalman filters where the Kalman gain and sensor and process noise characteristics are simultaneously determined. This is the first paper that jointly determines the Kalman gain and the robustness margin in terms of the process and sensor noise covariances. We consider continuous and discrete-time systems. The new results are presented in theorems 1, 2 and corollaries 1, 2. The paper also applies the proposed framework to two realistic aerospace problems -- orbit estimation of a satellite and state estimation of an F16 aircraft model.

\section{Technical Results}
\subsection{Discrete-Time Kalman Filtering}
\begin{theorem}\thmlab{thm1}
The largest process and sensor noise covariance for which the Kalman filter can meet a user-specified steady-state estimation error for a discrete-time dynamical system is given by the solution of the following convex optimization problem: 
\begin{subequations}
\begin{align}
\notag
&\min_{\K\in\real^{n\times p}, \vo{\eta} \in \real^m_+, \vo{\zeta}\in\real^p_+} \|\vo{\eta}\|_2 + \gamma\|\vo{\zeta}\|_\lambda\\ 
\notag &\text{ subject to } \\
&\begin{bmatrix}
\Sigss & (\I-\K\C)\A\sqrt{\Sigss} & (\I-\K\C)\B & \K\\
\ast   & \I & \vo{0} & \vo{0}\\
\ast   & \vo{0} & \diag{\vo{\eta}} & \vo{0}\\
\ast   & \ast   & \ast & \diag{\vo{\zeta}}
\end{bmatrix} \ge \vo{0},
\end{align}
\elab{thm1}
\end{subequations}
where $\gamma$ is used to weigh the relative importance of sensor and process noise, and $\lambda\in[1,2]$ defines the suitable norm for optimizing the sensor precision. The process and noise covariance are recovered as $\Q^{-1} := \diag{\vo{\eta}}$ and $\R^{-1} := \diag{\vo{\zeta}}$, and $\K$ is the Kalman gain.
\end{theorem}

\begin{proof}
For the system defined in \eqn{model}, the uncertainty propagation in the covariance is given by 
\begin{equation}
\Sigmab^-_{k+1} = \A\Sigmab^+_{k}\A^T + \B\Q\B^T,
\elab{cov_prop}
\end{equation}
where $\Sigmab^+_{k}$ is the posterior at time $k$ and $\Sigmab^-_{k+1}$ is the prior at time $k+1$. The Joseph form of the covariance update equation at time $k+1$ is given by 
\begin{equation}
\Sigmab^+_{k+1} = (\I-\K\C)\Sigmab^-_{k+1}(\I-\K\C)^T + \K\R\K^T.
\elab{cov_update}
\end{equation}
Combining, \eqn{cov_prop} and \eqn{cov_update}, we get
\begin{multline*}
\Sigmab^+_{k+1} = (\I-\K\C)\left(\A\Sigmab^+_{k}\A^T + \B\Q\B^T\right)  (\I-\K\C)^T \\+ \K\R\K^T.
\end{multline*}

In the steady-state when $\Sigmab^+_{k+1} = \Sigmab^+_{k} = \Sigmab_\infty$, we get
\begin{multline}
\Sigss = (\I-\K\C)\left(\A\Sigss\A^T + \B\Q\B^T\right)  (\I-\K\C)^T \\+ \K\R\K^T.
\elab{sscov}
\end{multline}
We next relax the equality condition in \eqn{sscov} \cite{willems1971least} to the following linear matrix inequality using Schur complement \cite{zhang2006schur},
\begin{align*}
\begin{bmatrix}
\Sigss & (\I-\K\C)\A\sqrt{\Sigss} & (\I-\K\C)\B & \K\\
\ast   & \I & \vo{0} & \vo{0}\\
\ast   & \vo{0} & \Q^{-1} & \vo{0}\\
\ast   & \ast   & \ast & \R^{-1}
\end{bmatrix} \ge \vo{0},
\end{align*}
which is convex in $\K,\Q^{-1}$ and $\R^{-1}$. The $\ast$ denotes symmetric terms. Noting that $\Q$ and $\R$ are diagonal, we define new variables $\vo{\eta} \in \real^m_+$ and $\vo{\zeta}\in\real^p_+$, such that $\Q^{-1} := \diag{\vo{\eta}}$ and $\R^{-1} := \diag{\vo{\zeta}}$. 
The largest process and sensor noise covariance for which the Kalman filter is able to meet the user-specified steady-state estimation error is achieved by minimizing the norms of $\vo{\eta}$ and $\vo{\zeta}$. The gain in the Kalman update equations is given by $\K$.
\end{proof}
\begin{corollary}\thmlab{Cor1}
Consider the case where the desired steady-state performance is specified as $\tr{\Sigss}\leq\theta$ instead of the exact steady-state estimation error. For this performance specification, Kalman gain with the optimal process and sensor noise covariances is obtained by solving the following optimization problem:
\begin{subequations}
\begin{align}
\notag &\min_{\W\in\real^{n\times p}, \vo{\eta} \in \real^m_+, \vo{\zeta}\in\real^p_+, \Z\in\real^{n\times n}_+, \X\in\real^{n\times n}_+} \|\vo{\eta}\|_2 + \gamma\|\vo{\zeta}\|_\lambda\\
\notag &\text{ subject to } \\
&\begin{bmatrix}
\Z & (\Z-\W\C)\A & (\Z-\W\C)\B & \W\\
\ast   & \Z & \vo{0} & \vo{0}\\
\ast   & \ast & \diag{\vo{\eta}} & \vo{0}\\
\ast   & \ast   & \ast & \diag{\vo{\zeta}}
\end{bmatrix} \ge \vo{0},\elab{cor1A}\\
&\begin{bmatrix}
\X  & \I\\
\ast & \Z\\
\end{bmatrix} \ge \vo{0}, \elab{cor1B}\\
& \tr{\X}\leq \theta \elab{cor1C}
\end{align}
\end{subequations}
where Kalman gain and steady-state error covariance are obtained as $\K := \Z^{-1}\W$ and $\Sigss := \Z^{-1}$ respectively.The process and noise covariance are recovered as $\Q^{-1} := \diag{\vo{\eta}}$ and $\R^{-1} := \diag{\vo{\zeta}}$.  
\end{corollary}
\begin{proof}
Starting with the relaxation of \eqn{sscov}, we arrive at the inequality condition,
\begin{multline}
\Sigss - (\I-\K\C)\left(\A\Sigss\A^T\right)  (\I-\K\C)^T \\- (\I-\K\C)\left(\B\Q\B^T\right)  (\I-\K\C)^T - \K\R\K^T \ge \vo{0}.
\elab{sscov2}
\end{multline}
Next, we perform congruent transform with $\Z:=\Sigss^{-1}$, and substitute $\W:=\Z\K$ in \eqn{sscov2} to obtain,
\begin{multline}
\Z - (\Z-\W\C)\left(\A\Z^{-1}\A^T\right)(\Z-\W\C)^T \\- (\Z-\W\C)\left(\B\Q\B^T\right)  (\Z-\W\C)^T - \W\R\W^T \ge \vo{0}.
\elab{sscov3}
\end{multline}
Applying Schur complement on \eqn{sscov3} and substituting $\Q^{-1} := \diag{\vo{\eta}}$ and $\R^{-1} := \diag{\vo{\zeta}}$, we get
\begin{align*}
\begin{bmatrix}
\Z & (\Z-\W\C)\A & (\Z-\W\C)\B & \W\\
\ast   & \Z & \vo{0} & \vo{0}\\
\ast   & \ast & \diag{\vo{\eta}} & \vo{0}\\
\ast   & \ast   & \ast & \diag{\vo{\zeta}}
\end{bmatrix} \ge \vo{0}.
\end{align*}
The constraint $\tr{\Sigss}\leq\theta$ is equivalent to $\X - \Z^{-1} \geq \vo{0}, \tr{\X}\leq\theta$. By applying Schur complement, this can be written as a linear matrix inequality,
\begin{align*}
&\begin{bmatrix}
X  & \I\\
\ast & \Z\\
\end{bmatrix} \ge \vo{0}.
\end{align*}
Minimizing the norms of $\vo{\eta}$ and $\vo{\zeta}$ subject to \eqn{cor1A}, \eqn{cor1B} and \eqn{cor1C} achieve the largest process and sensor noise covariance satisfying $\tr{\Sigss}\leq\theta$. 
\end{proof}

\begin{remark}{\bf Sparse Sensing:}
With overlapping sensing modalities, minimizing the $\|\vo{\zeta}\|_{1}$ norm may lead to a sparse $\vo{\zeta}$ vector, which in turn results in a Kalman filtering process with a reduced set of active sensors. If the Kalman gain $\K$ can achieve the desired steady-state accuracy even when some sensors have $\zeta_i = 0$ (or infinite noise covariance), it implies that those sensors do not contribute to the update law. Consequently, the corresponding column of the optimal $\vo{K}$ matrix will be zero.
\end{remark}

\begin{remark}{\bf Quantifying Largest Tolerable Sensor Degradation:}
Minimizing the $\|\vo{\zeta}\|_{2}$ norm is recommended to optimize a sensing architecture without necessarily focusing on sparsity. This approach leads to the largest sensor noise covariance achievable for a given sensor set. Practically, this means installing sensors with precisions higher than the $l_2$ minimum $\vo{\zeta}$. The difference between the minimum $\vo{\zeta}$  and the installed sensors' actual (higher) precision of the installed sensors can then be used to quantify the permissible degradation in sensor quality. 
\end{remark}

\subsection{Continuous-Time Kalman-Bucy Filtering}
In this section, consider the continuous-time system
\begin{subequations}
    \begin{align}
    \text{State Dynamics:} && \xd(t) &= \A\x(t) + \B\w(t),\\
    \text{Measurement Model:}&&\y(t) &= \C\x(t) + \n(t),
    \end{align}
    \elab{modelct}
\end{subequations}
where $\x(t)\in\real^n$ is the state vector, $\w(t)\in\real^m$ is the process noise, $\y(t)\in\real^p$ is the measurement data, and $\n(t)$ is the sensor noise.  We assume the process and the sensor noise to be zero mean Gaussian, i.e.,  $\Exp{\w(t)} = \vo{0}$, $\Exp{\n(t)} = \vo{0}$, $\Exp{\w(t)\vo{w}^T(t)} = \Q$, and $\Exp{\n(t)\vo{n}^T(t)} = \R$.

The differential equation for the error covariance is given by
$$
\dot{\Sigmab} = (\A-\K\C)\Sigmab + \Sigmab(\A-\K\C)^T + \K\R\K^T + \B\Q\B^T, 
$$
which in the steady state becomes
\begin{equation}
    (\A-\K\C)\Sigss + \Sigss(\A-\K\C)^T + \K\R\K^T + \B\Q\B^T = \vo{0}.
    \elab{Contss}
\end{equation}

\begin{theorem}
    The largest process and sensor noise covariance for which the Kalman filter can meet a user-specified steady-state estimation error for a continuous-time dynamical system is given by the solution of the following convex optimization problem:

    \begin{subequations}
        \begin{align}
        \notag &\min_{\K\in\real^{n\times p}, \vo{\eta} \in \real^m_+, \vo{\zeta}\in\real^p_+} \|\vo{\eta}\|_2 + \gamma\|\vo{\zeta}\|_\lambda\\ 
        \notag &\text{ subject to } \\
        &\begin{bmatrix}
            \sym{(\A-\K\C)\Sigss} & \B & \K\\
        \ast   &-\diag{\vo{\eta}} & \vo{0}\\
        \ast   & \ast & -\diag{\vo{\zeta}}
        \end{bmatrix} \le \vo{0},
        \end{align}
        \elab{thm2}
        \end{subequations}
        where $\gamma$ is used to weigh the relative importance of sensor and process noise, $\lambda\in[1,2]$ defines the suitable norm for optimizing the sensor precisions, and $\sym{(\cdot)} := (\cdot) + (\cdot)^T$. The process and noise covariance are recovered as $\Q^{-1} := \diag{\vo{\eta}}$ and $\R^{-1} := \diag{\vo{\zeta}}$, and $\K$ is the Kalman gain.
        \thmlab{thm2}
\end{theorem}
\begin{proof}
    Starting with the inequality constraint obtained by the relaxation of \eqn{Contss} 
    \begin{equation}
    (\A-\K\C)\Sigss + \Sigss(\A-\K\C)^T + \K\R\K^T + \B\Q\B^T \leq \vo{0},
    \elab{Contssineq}
    \end{equation}
    we apply Schur complement to arrive at the following LMI,
    \begin{equation}
        \begin{bmatrix}
            \sym{(\A-\K\C)\Sigss} & \B & \K\\
            \ast   &-\Q^{-1} & \vo{0}\\
            \ast   & \ast & -\R^{-1}
        \end{bmatrix} \le \vo{0},
        \elab{T2Const}
    \end{equation}
    Substitutions of $\Q^{-1} := \diag{\vo{\eta}}$ and $\R^{-1} := \diag{\vo{\zeta}}$ in the inequality \eqn{T2Const} leads to the constraint in \eqn{thm2}.

\end{proof}
\begin{corollary}\thmlab{Cor2}
In the case where the target steady-state performance is defined as $\tr{\Sigss}\leq\theta$ rather than a precise steady-state estimation error, the Kalman gain with the optimal process and sensor noise covariances is determined by solving the subsequent optimization problem,
\begin{subequations}
\begin{align}
    \notag &\min_{\W\in\real^{n\times p}, \vo{\eta} \in \real^m_+, \vo{\zeta}\in\real^p_+, \Z\in\real^{n\times n}_+ } \|\vo{\eta}\|_2 + \gamma\|\vo{\zeta}\|_\lambda\\
    \notag &\text{ subject to } \\
    & \begin{bmatrix}
    \sym{\Z\A-\W\C} & \Z\B & \W\\
    \ast   &-\diag{\vo{\eta}} & \vo{0}\\
    \ast   & \ast & -\diag{\vo{\zeta}}
    \end{bmatrix} \le \vo{0},\\
    &\begin{bmatrix}
    \X  & \I\\
    \ast & \Z\\
    \end{bmatrix} \ge \vo{0},\\
    & \tr{\X}\leq \theta
\end{align}
\elab{Cor2}
\end{subequations}
where the solution is given by $\Q^{-1} := \diag{\vo{\eta}}$, $\R^{-1} := \diag{\vo{\zeta}}$, $\Sigss := \Z^{-1}$, $\K := \Z^{-1}\W$.  
\end{corollary}
\begin{proof}
The proof follows similar steps as corollary \thm{Cor1} starting with the relaxation of \eqn{Contss}. Further steps are omitted due to space restrictions.
\end{proof}

\section{Applications}
This section applies the proposed theorems and corollaries to sensing problems in spacecraft proximity operation in low earth orbit (LEO) and longitudinal flight dynamics for F-16 aircraft.  
\subsection{Spacecraft rendezvous maneuver}
The Clohessy–Wiltshire-Hill(CWH) equations \cite{CWEqn1,Hill1878} provide a simplified representation of orbital relative motion, assuming a circular orbit for the target and an elliptical or circular orbit for the chaser spacecraft. This model offers a first-order approximation of the chaser's motion within a target-centered coordinate system and is frequently used in planning rendezvous maneuvers in Low Earth Orbit(LEO) and formation flight dynamics \cite{CWEqn2017}. The CHW model with full state measurement is described by the following linear time-invariant (LTI) system:
\begin{subequations}
    \begin{align}
    \dot{\x}(t) &= \A\x(t) + \B_u\vo{u}(t) + \B_w\w(t),\\
    \y(t) &= \C\x(t) + \n(t),
    \end{align}
    \elab{CHWMod}
\end{subequations}
where
\begin{equation}
\begin{aligned}
\notag
\A &= {\begin{bmatrix}
    0 & 0 & 0 & 1 & 0 & 0\\
    0 & 0 & 0 & 0 & 1 & 0\\
    0 & 0 & 0 & 0 & 0 & 1\\
    3\oref^2 & 0 & 0 & 0 & 2\oref & 0\\
    0 & 0 & 0 & -2\oref & 0 & 0\\
    0 & 0 & -\oref^2 & 0 & 0 & 0
\end{bmatrix}},\\
\notag
\B_u &= {\begin{bmatrix}
    0 & 0 & 0 & 1 & 0 & 0\\
    0 & 0 & 0 & 0 & 1 & 0\\
    0 & 0 & 0 & 0 & 0 & 1\\
\end{bmatrix}}^T,\\
\notag
\B_d &= B_u,\\
\notag
\C &= \I^{n\times n},
\end{aligned}
\end{equation}

$\x := [x\,y\,z\,\dot{x}\,\dot{y}\,\dot{z}]'$, $\vo{u} := [F_x/m\,F_y/m\,F_z/m]'$ are state and control vectors respectively. Constant $\oref:= \sqrt{\mu/a_{ref}}$ corresponds to the angular velocity of the target spacecraft in the circular orbit where $\mu$ is the earth's gravitational parameter and $a_{ref}$ is orbital radius of the target spacecraft. For a LEO target, $\oref$ is approximately equal to $0.00113$ $rad/s$. The state of the CHW model comprises the relative position and velocities along different directions:
\begin{itemize}
    \item Along the radial direction of the target spacecraft motion: $(x,\dot{x}),$
    \item Tangential to the target satellite orbit: $(y,\dot{y}),$
    \item Direction normal to the orbital plane: $(z,\dot{z}).$
\end{itemize}
Impulsive thrust along respective directions $\vo{F}:= [F_x\,F_y\,F_z]'$ divided by the mass of the chaser spacecraft $m$ acts as the control input. We assume small perturbation in thrusters as the source of process noise $\vo{w}$. The sensor noise $\vo{n}$ corrupts the full-state measurement model. The discrete-time model is obtained by discretizing the continuous model \eqn{CHWMod} using the Tustin method \cite{Tustin2022} with sampling time $T_s = 0.01$ sec.

Simulation studies are performed on two cases for discrete and continuous time systems. The cost functions considered for the optimization problems in corollaries \thm{Cor1} and \thm{Cor2} for each case are
\begin{equation}
    J = 
    \begin{cases}
         \|\vo{\eta}\|_2 + \gamma\|\vo{\zeta}\|_2 & \text{Case 1 (c1)}\\
         \|\vo{\W_q\eta}\|_2 + \gamma\|\vo{W_r\zeta}\|_2 & \text{Case 2 (c2)}\\
    \end{cases}
    \elab{Jcase}
\end{equation}
with $\gamma = 1$. Case 1 assigns equal weight to process and sensor noise variances during optimization. Internal weights $W_q = \diag{\begin{bmatrix}1&100&10\end{bmatrix}}$ and $W_r = \diag{\begin{bmatrix}100 & 10 & 1 & 100 & 10 & 1\end{bmatrix}}$ in case 2 adds varying priorities to each noise with the expectation of increment in optimal variance for those linked with larger weight. Both cases are constrained by steady-state filter performance criteria, $\tr{\Sigss} \leq 0.1$.

We arbitrarily consider the continuous time model with parameters in case 2 to simulate real-time filter behavior with process and sensor noise characteristics and Kalman gain obtained from the optimization in corollary \thm{Cor2}. The plots of mean and variance propagation for relative position errors $(e_x,e_y,e_z)$ and velocity errors $(e_{\dot{x}},e_{\dot{y}},e_{\dot{z}})$ respectively presented in \fig{fig:posErr} and \fig{fig:velErr}, demonstrate convergence of the filter. 

Next, we study the effect of internal weights $\W_q$, $\W_r$ on process and sensor noise variances by comparing optimal outcomes for cases 1 and 2. \fig{fig:CHWcont} and \fig{fig:CHWdisc} show bar plot comparisons between both cases for continuous and discrete time models respectively. As expected from weight allocation, variance of noises $(w_{F_y}, n_x, n_{\dot{x}})$ in case 2 are larger than case 1. 

\begin{figure}[h]
\centering
\includegraphics[width=0.9\linewidth]{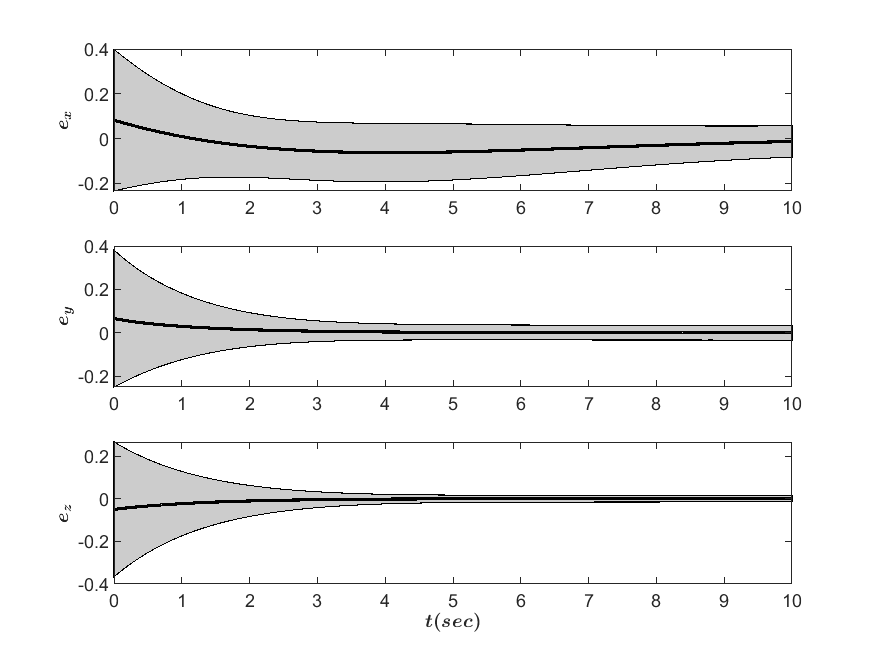}
\caption{Positional error evolution. Solid lines denote mean error in relative position $(x,y,z)$, whereas shaded regions show $\mu_i \pm \sqrt{\Sigma_{ii}}$ where $i \in (e_x, e_y, e_z)$}
\flab{fig:posErr}
\end{figure}

\begin{figure}[h]
\centering
\includegraphics[width=0.9\linewidth]{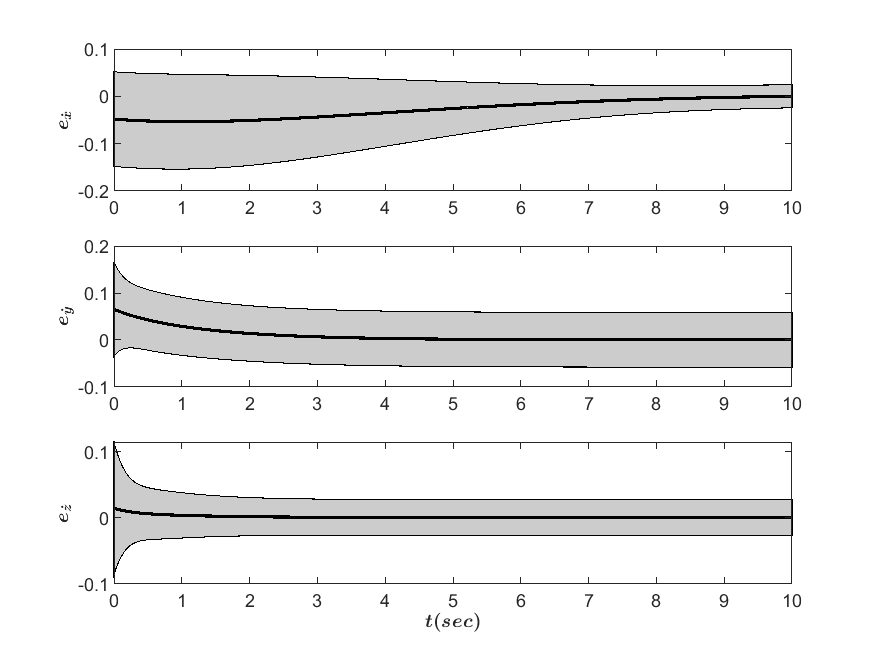}
\caption{Velocity error evolution. Solid lines denote mean error in relative velocity $(\dot{x},\dot{y},\dot{z})$, whereas shaded regions show $\mu_i \pm \sqrt{\Sigma_{ii}}$ where $i \in (e_{\dot{x}}, e_{\dot{y}}, e_{\dot{z}})$}
\flab{fig:velErr}
\end{figure}

\begin{figure}[h]
\centering
\includegraphics[width=0.7\linewidth]{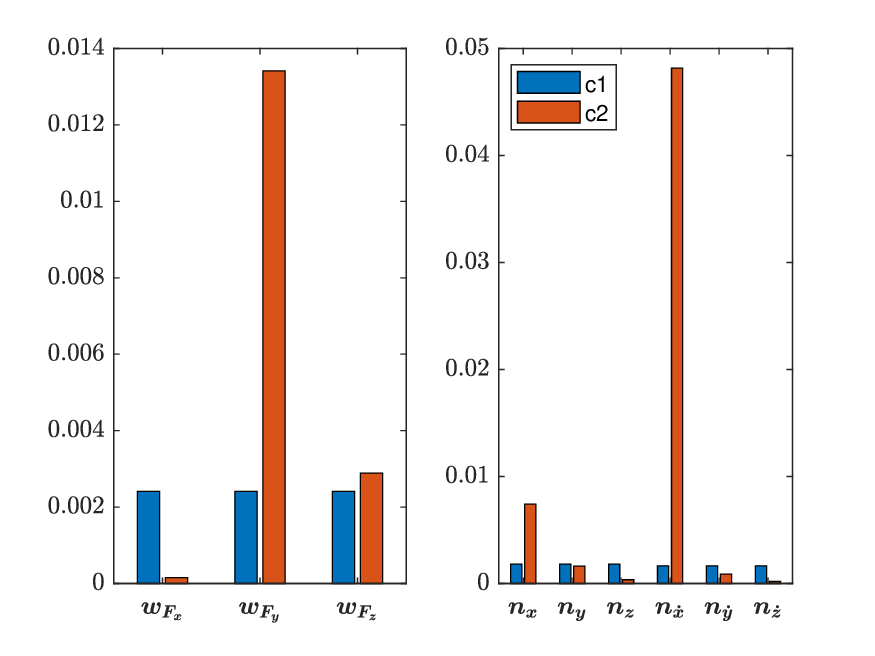}
\caption{Bar plot comparing process(left) and sensor(right) noise variances between case 1(c1) and case 2(c2) for the continuous-time system. The units for $(w_{F_x}, w_{F_y}, w_{F_z})$, $(n_x,n_y,n_z)$, and $(n_{\dot{x}}, n_{\dot{y}}, n_{\dot{z}})$ are $(m/s^2)^2$, $m^2$, and $(m/s)^2$ respectively.}
\flab{fig:CHWcont}
\end{figure}

\begin{figure}[h]
\centering
\includegraphics[width=0.7\linewidth]{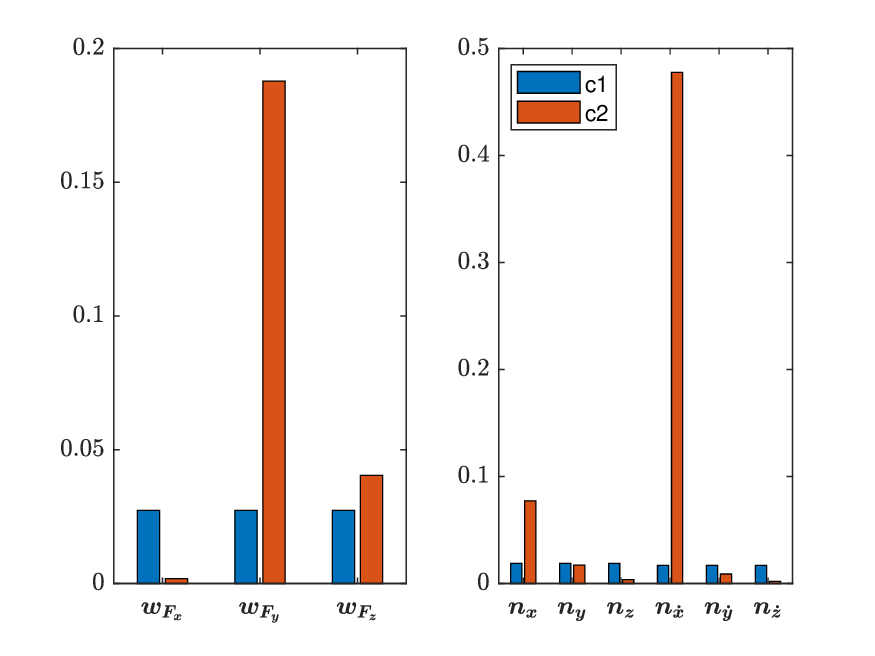}
\caption{Comparison of process(left) and sensor(right) noise variances between case 1(c1) and case 2(c2) for the discrete-time system. The units for $(w_{F_x}, w_{F_y}, w_{F_z})$, $(n_x,n_y,n_z)$, and $(n_{\dot{x}}, n_{\dot{y}}, n_{\dot{z}})$ are $(m/s^2)^2$, $m^2$, and $(m/s)^2$ respectively.}
\flab{fig:CHWdisc}
\end{figure}

\subsection{Flight control}
Let us consider the longitudinal dynamics of F-16 aircraft linearized about the equilibrium/trim points
\begin{equation}\label{eq_point}
\begin{aligned}
\notag
\x_0 &= {\begin{bmatrix}
    1000 && -3.02\times10^{-3} && -3.02\times10^{-3} && 0
\end{bmatrix}}^T,\\
\notag
\vo{u}_0 &= {\begin{bmatrix}
    6041.2 & -1.38
\end{bmatrix}}^T,\\
\end{aligned}
\end{equation}
for steady flight conditions. The LTI system for linearized longitudinal dynamics is given by\cite{SLJ2015}
\begin{subequations}
    \begin{align}
    \dot{\x}(t) &= \A\x(t) + \B_w\w(t),\\
    \y(t) &= \C_y\x(t) + \n(t),
    \end{align}
    \elab{F16Mod}
\end{subequations}
where
\begin{equation}
\begin{aligned}
\notag
\A &= {\begin{bmatrix}
    -1.8969\times10^{-2} & -0.4052 & -32.17 & 0.8915\\
    -6.4397\times10^{-5} & -1.6176 &    0   & 0.9325\\
               0         &     0   &    0   &    1  \\
               0         & -2.3683 &    0   &-1.9696
\end{bmatrix}},\\
\notag
\B_d &= {\begin{bmatrix}
            1 & 0 & 0 & 0\\
            0 & 1 & 0 & 0\\
            0 & 0 & 0 & 1\\
        \end{bmatrix}}^T,\\
\notag
\C_y &= \begin{bmatrix}
    -0.0191 & -5.2893 & -32.17 & 3.7071\\
    -0.0643 & -1.6176 & 0.0971 & 932.5332\\
       0    &     1   &    0   &    0    \\
       0    &     0   &    0   &    1    \\
     1.7578 &     0   &    0   &    0    
\end{bmatrix}.
\end{aligned}
\end{equation}

The state vector is denoted as $\x=[V\ \alpha\ \theta\ q]'$, representing the velocity $V$ in $ft/s$, angle of attack $\alpha$ in \textit{radians}, pitch angle $\theta$ in \textit{radians}, and pitch rate $q$ in $rad/s$. The measurement vector is $y=[\dot{u}\ \dot{w}\ \alpha\ q\ \bar{q}]'$, including body acceleration along the roll axis $\dot{u}$ in $ft/s^2$, body acceleration along the yaw axis $\dot{w}$ in $ft/s^2$, angle of attack $\alpha$ in \textit{radians}, pitch rate $q$ in \textit{rad/s}, and dynamic pressure $\bar{q}:= \rho_{atm}V^2/2$ in $lb/ft^2$, where $\rho_{atm}$ defines atmospheric density at an altitude corresponding to the steady flight.
Similar to the previous example, we consider two cases for the continuous time system where cost functions are given by \eqn{Jcase}. In case 1(c1), where $\gamma = 1$, equal weightage is assigned to both process and sensor noise variances during optimization. Case 2(c2) assigns internal weights $W_q = \diag{\begin{bmatrix}1&10&1\end{bmatrix}}$ to process noise and $W_r = \diag{\begin{bmatrix}1 & 1 & 0.1 & 1 & 1\end{bmatrix}}$ to sensor noise. The steady-state filter performance constraint $\tr{\Sigss} \leq 0.1$ is imposed on both cases.

We investigate the influence of internal weights, denoted as $\W_q$ and $\W_r$, on the process and sensor noise variances by comparing the optimal results between cases 1 and 2. \fig{fig:F16cont} shows bar plot comparison between both cases. As anticipated from the allocation of weights, there is an observed increase in the variance of the process noise $w_\alpha$ and a decrease in the variance of the sensor noise $n_\alpha$ in case 2. The internal weights $\W_q$ and $\W_r$ shape the optimization outcome, influencing the tradeoff between process and sensor noise variances. This optimization framework facilitates adding constraints that explicitly specify lower (upper) bounds on the variance (precision) of the noises, i.e.
$$
\eta_i \leq \eta_{max}, \zeta_i \leq \zeta_{max}.
$$
Adding these constraints may lead to infeasibility, which can be addressed by tweaking the performance constraints.

Next, we consider the case of sparse sensing, which is simulated by using the cost function $J:= \|\vo{\eta}\|_2 + \gamma\|\vo{\zeta}\|_1$ with $\gamma=1$. Stem plot of sensor precisions $\zeta_i$ in \fig{fig:F16sparse} suggests a sparse sensor configuration by eliminating sensor requirement for $\alpha$ and $q$. In some cases, we may obtain significantly small non-zero sensor precision for specific sensors. Iteratively reweighted minimization proposed in \cite{VD2021} can be implemented alongside the proposed algorithms for such cases to obtain sparse sensing architectures while maximizing process noise variance. 

\begin{figure}[h]
\centering
\includegraphics[width=0.7\linewidth]{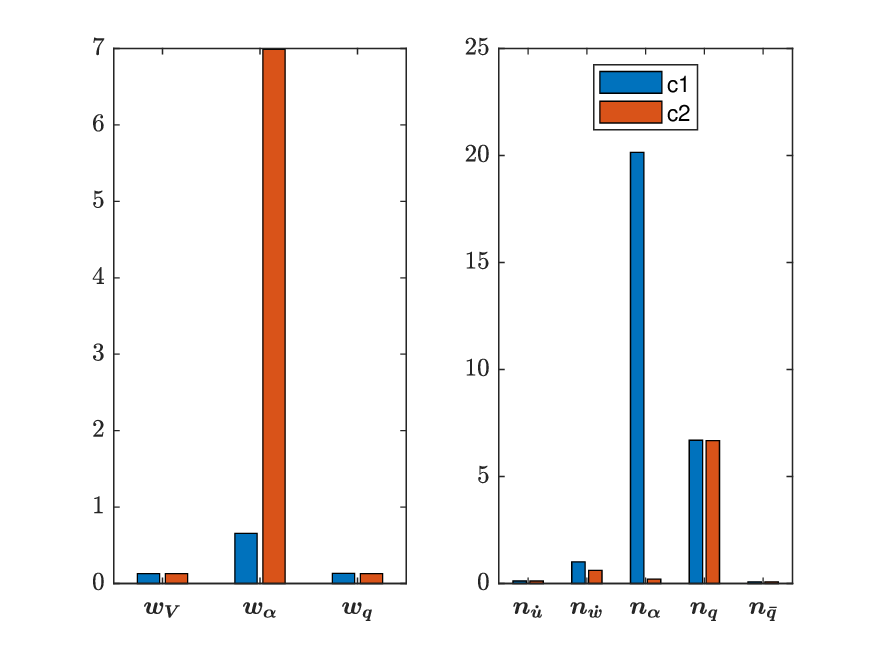}
\caption{Bar plot comparing process(left) and sensor(right) noise variances between case 1(c1) and case 2(c2) for the longitudinal flight dynamics. The units for $w_V$, $(w_\alpha,n_\alpha)$, $(w_q,n_q)$, $(n_{\dot{u}}, n_{\dot{w}})$ and $n_{\bar{q}}$ are $(ft/s)^2$, $rad^2$, $(rad/s)^2$, $(ft/s^2)^2$, and $(lb/ft^2)^2$ respectively.}
\flab{fig:F16cont}
\end{figure}

\begin{figure}[h]
\centering
\includegraphics[width=0.7\linewidth]{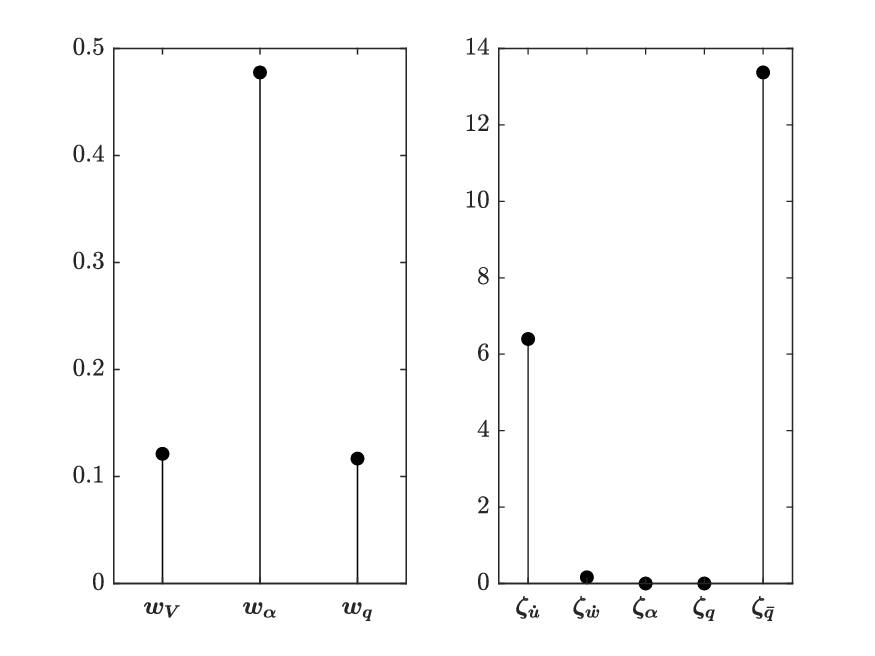}
\caption{Stem plot of process noise variances(left) and sensor precisions(right). The units for $w_V$, $w_\alpha$, $w_q$, $\zeta_\alpha$, $\zeta_q$, $(\zeta_{\dot{u}}, \zeta_{\dot{w}})$ and $\zeta_{\bar{q}}$ are $(ft/s)^2$, $rad^2$, $(rad/s)^2$,$rad^{-2}$, $(rad/s)^{-2}$,  $(ft/s^2)^{-2}$, and $(lb/ft^2)^{-2}$ respectively.}
\flab{fig:F16sparse}
\end{figure}

\section{Summary \& Conclusion}
This paper introduced a novel convex optimization framework maximizing the robustness margin for the Kalman filter while ensuring adherence to a predefined steady-state error budget. The formulation is detailed for both discrete and continuous linear time-varying dynamics. The proposed algorithms can yield an optimal sensing architecture with or without sparseness while maximizing the process noise variance. This is demonstrated using the examples of spacecraft rendezvous maneuvers and longitudinal dynamics of F-16 aircraft. 

\section{Acknowledgements}
This research is sponsored by the AFOSR Dynamic Data and Information Processing Program, grant FA9550-22-1-0539.

\bibliographystyle{unsrt}
\bibliography{root}

\end{document}

%% file: macros.tex
\usepackage{tikz}
\usepackage{cite}
\usepackage{amsmath,amssymb,amsfonts,amsthm}
\usepackage{algorithmic}
\usepackage{graphicx}
\usepackage{textcomp}
\usepackage{xcolor}
\usepackage{mathtools}

\usepackage[font={small,it}, labelfont=bf]{caption}
\usetikzlibrary{automata, positioning, arrows}

\newcommand{\vo}[1]{\boldsymbol{#1}}
\newcommand{\set}[1]{\mathbb{#1}}
\newcommand{\real}{\set{R}}

\newcommand{\x}{\vo{x}}
\newcommand{\xd}{\dot{\vo{x}}}
\newcommand{\y}{\vo{y}}

\newcommand{\n}{\vo{n}}
\newcommand{\w}{\vo{w}}

\newcommand{\X}{\vo{X}}

\newcommand{\Z}{\vo{Z}}
\newcommand{\A}{\vo{A}}
\newcommand{\B}{\vo{B}}
\newcommand{\C}{\vo{C}}

\newcommand{\I}{\vo{I}}
\newcommand{\Q}{\vo{Q}}
\newcommand{\R}{\vo{R}}
\newcommand{\K}{\vo{K}}
\newcommand{\W}{\vo{W}}
\newcommand{\tr}[1]{\textbf{tr}\left[#1\right]}
\newcommand{\diag}[1]{\textbf{diag}\left(#1\right)}

\newcommand{\sym}[1]{\textbf{sym}\left(#1\right)}

\newcommand{\Sigmab}{\vo{\Sigma}}
\newcommand{\Sigss}{\vo{\Sigma}_\infty}

\newcommand{\oref}{\omega_{ref}}

\newcommand{\Exp}[1]{\mathbb{E}\left[#1\right]}

\newcommand{\elab}[1]{\label{eqn:#1}}
\newcommand{\eqn}[1]{(\ref{eqn:#1})}

\newcommand{\flab}[1]{\label{fig:#1}}
\newcommand{\fig}[1]{Fig.\ref{fig:#1}}

\newcommand{\thmlab}[1]{\label{thm:#1}}
\newcommand{\thm}[1]{\ref{thm:#1}}

\newtheorem{theorem}{Theorem}
\newtheorem{remark}{Remark}
\newtheorem{corollary}{Corollary}